\documentclass[runningheads]{llncs}
\usepackage{xspace}
\usepackage{xparse}
\usepackage{xifthen}
\usepackage{tabto}
\usepackage{scalerel}
\usepackage{stackengine}
\usepackage{lscape}

\usepackage{booktabs}
\usepackage{tabularx}
\usepackage{makecell}
\usepackage{paralist}
\def\rot{\rotatebox}

\usepackage{amsmath}
\usepackage{amsfonts}
\usepackage{amssymb}
\usepackage{mathrsfs}
\usepackage{cancel}

\usepackage{graphicx}
\usepackage{algorithmicx}
\usepackage{algorithm}
\usepackage{algpseudocode}

\usepackage{cleveref}
    \crefname{figure}{Figure}{Figures}
    \Crefname{figure}{Figure}{Figures}
    \crefname{section}{Section}{Sections}
    \Crefname{section}{Section}{Sections}
    \crefname{appendix}{Appendix}{Appendices}
    \Crefname{appendix}{Appendix}{Appendices}
    \crefname{definition}{Definition}{Definitions}
    \Crefname{definition}{Definition}{Definitions}
    \crefname{notation}{Notation}{Notations}
    \Crefname{notation}{Notation}{Notations}
    \crefname{conjecture}{Conjecture}{Conjectures}
    \Crefname{conjecture}{Conjecture}{Conjectures}
    \crefname{remark}{Remark}{Remarks}
    \Crefname{remark}{Remark}{Remarks}
    \crefname{lemma}{Lemma}{Lemmas}
    \Crefname{lemma}{Lemma}{Lemmas}
    \crefname{lresult}{Result}{Results}
    \Crefname{lresult}{Result}{Results}
    \crefname{theorem}{Theorem}{Theorems}
    \Crefname{theorem}{Theorem}{Theorems}
    \crefname{corollary}{Corollary}{Corollaries}
    \Crefname{corollary}{Corollary}{Corollaries}
    \crefname{equation}{Equation}{Equations}
    \Crefname{equation}{Equation}{Equations}
    \crefname{algorithm}{Algorithm}{Algorithms}
    \Crefname{algorithm}{Algorithm}{Algorithms}
    \Crefname{line}{Line}{Lines}
    \crefname{line}{Line}{Lines}   

\NewDocumentCommand\rp{m}{{\left(#1\right)}}
\NewDocumentCommand\qp{m}{{\left[#1\right]}}

\NewDocumentCommand\ap{m}{{\left\langle#1\right\rangle}}

\NewDocumentCommand\equivalent{o}{
    \IfValueTF{#1} 
    {{\overset{#1}{\sim}}}
    {{\sim}}
}

\NewDocumentCommand\sign{o}{
    \IfValueTF{#1}
    {{\text{sign}\rp{#1}}}
    {{\text{sign}}}
}

\stackMath
\NewDocumentCommand\pagoda{m}{
    \savestack{\tmpbox}{\stretchto{
        \scaleto{
            \scalerel*[\widthof{\ensuremath{#1}}]{\kern+1.5pt\bigwedge\kern+1.5pt}
            {\rule[-\textheight/2]{1ex}{\textheight}}
        }{\textheight}
    }{0.5ex}}
    \stackon[1pt]{#1}{\tmpbox}
}
\makeatletter
\newcommand\xleftrightarrow[2][]{%
  \ext@arrow 9999{\longleftrightarrowfill@}{#1}{#2}}
\newcommand\longleftrightarrowfill@{%
  \arrowfill@\leftarrow\relbar\rightarrow}
\makeatother

\NewDocumentCommand\probability{}{\mathcal{P}}
\NewDocumentCommand\probabilityfunc{om}{
    \IfValueTF{#1}
    {{\probability#1\qp{#2}}}
    {{\probability\qp{#2}}}
}

\NewDocumentCommand\powerset{om}{
    \IfValueTF{#1}
    {{\mathbb{P}^{#1}\rp{#2}}}
    {{\mathbb{P}\rp{#2}}}
}

\NewDocumentCommand\bin{mmo}{
    \IfValueTF{#3}
    {{\text{Bin}\qp{#1, #2}\rp{#3}}}
    {{\text{Bin}\qp{#1, #2}}}
}

\NewDocumentCommand\pois{mo}{
    \IfValueTF{#2}
    {{\text{Poisson}\qp{#1}\rp{#2}}}
    {{\text{Poisson}\qp{#1}}}
}

\NewDocumentCommand\bern{mo}{
    \IfValueTF{#2}
    {{\text{Bern}\qp{#1}\rp{#2}}}
    {{\text{Bern}\qp{#1}}}
}

\NewDocumentCommand\bcoin{}{BroadcastCoin\xspace}
\NewDocumentCommand\apost{}{ConsensusOnDemand\xspace}

\NewDocumentCommand{\event}{mmo}{
    \IfValueTF{#3}
    {{\ap{#1.\textrm{#2} \mid #3}}}
    {{\ap{#1.\textrm{#2}}}}
}

\algnewcommand\Instance[2]{\State #1, \textbf{instance} #2}
\algnewcommand\Instances[2]{\State #1, \textbf{instances} #2}
\algnewcommand\MultipleInstances[1]{\State #1, \textbf{multiple instances}}
\algnewcommand\Exposes[3]{\State #1 \textbf{exposes} #2 \textbf{as} #3}

\algnewcommand\Alias[2]{\State \textbf{use} $#1$ \textbf{as alias for} $#2$}

\algnewcommand\Trigger[3]{\State \textbf{trigger} $\event{#1}{#2}[#3]$}
\algnewcommand\TriggerPure[2]{\State \textbf{trigger} $\event{#1}{#2}$}
\algnewcommand\SetTimeout[2]{\State \textbf{set timeout} $#1$ \textbf{in} $#2$}

\algsetblockdefx[Implements]{Implements}{EndImplements}{}{}{\textbf{Implements:}}{}
\algsetblockdefx[Uses]{Uses}{EndUses}{}{}{\textbf{Uses:}}{}
\algsetblockdefx[Parameters]{Parameters}{EndParameters}{}{}{\textbf{Parameters:}}{}
\algsetblockdefx[Optimizations]{Optimizations}{EndOptimizations}{}{}{\textbf{Optimizations:}}{}

\algsetblockdefx[Upon]{Upon}{EndUpon}{}{}[3]{\textbf{upon event} $\event{#1}{#2}[#3]$ \textbf{do}}{}
\algsetblockdefx[UponSuchThat]{UponSuchThat}{EndUponSuchThat}{}{}[4]{\textbf{upon event} $\event{#1}{#2}[#3]$ \textbf{such that} $#4$ \textbf{do}}{}

\algsetblockdefx[UponPure]{UponPure}{EndUponPure}{}{}[2]{\textbf{upon event} $\event{#1}{#2}$ \textbf{do}}{}
\algsetblockdefx[UponTimeout]{UponTimeout}{EndUponTimeout}{}{}[1]{\textbf{upon timeout} $#1$ \textbf{do}}{}
\algsetblockdefx[UponExists]{UponExists}{EndUponExists}{}{}[2]{\textbf{upon exists} $#1$ \textbf{such that} $#2$ \textbf{do}}{}
\algsetblockdefx[UponCondition]{UponCondition}{EndUponCondition}{}{}[1]{\textbf{upon} $#1$ \textbf{do}}

\algsetblockdefx[Procedure]{Procedure}{EndProcedure}{}{}[2]{\textbf{procedure} \emph{#1}($#2$) \textbf{is}}{}

\algsetblockdefx[IfExists]{IfExists}{EndIfExists}{}{}[2]{\textbf{if exists} $#1$ \textbf{such that} $#2$ \textbf{then}}{\textbf{end if}}
\algsetblockdefx[If]{If}{EndIf}{}{}[1]{\textbf{if} $#1$ \textbf{then}}{\textbf{end if}}

\algsetblockdefx[While]{While}{EndWhile}{}{}[1]{\textbf{while} $#1$ \textbf{do}}{\textbf{end while}}
\algdef{SE}[DoUntil]{DoUntil}{EndDoUntil}{\textbf{do}}[1]{\textbf{until} $#1$}
\algsetblockdefx[For]{For}{EndFor}{}{}[3]{\textbf{for} $#1 \;\textbf{in}\; #2..#3$ \textbf{do}}{\textbf{end for}}
\algsetblockdefx[ForAll]{ForAll}{EndForAll}{}{}[2]{\textbf{for all} $#1 \;\textbf{in}\; #2$ \textbf{do}}{\textbf{end for}}
\algsetblockdefx[ForAllSuchThat]{ForAllSuchThat}{EndForAllSuchThat}{}{}[3]{\textbf{for all} $#1 \;\textbf{in}\; #2$ \textbf{such that} $#3$ \textbf{do}}{\textbf{end for}}
\algsetblockdefx[NTimes]{NTimes}{EndNTimes}{}{}[1]{\textbf{for} $#1$ \textbf{times do}}{\textbf{end for}}

\makeatletter
\newcommand{\@chapapp}{\relax}%
\makeatother

\usepackage[title,toc,titletoc,header]{appendix}
\usepackage{csquotes}
\usepackage{dirtytalk}
\usepackage{wrapfig}

\usepackage[T1]{fontenc}
\usepackage{graphicx}
\begin{document}
\title{Consensus on Demand}
%

\author{Jakub Sliwinski \and
Yann Vonlanthen \and
Roger Wattenhofer }


\institute{ETH Zürich}
\authorrunning{J. Sliwinski et al.}

%
%
\maketitle              
\begin{abstract}
Digital money can be implemented efficiently by avoiding consensus. However, no-consensus designs are fundamentally limited, as they cannot support general smart contracts, and similarly they cannot deal with conflicting transactions.

We present a novel protocol that combines the benefits of an asynchronous, broadcast-based digital currency, with the capacity to perform consensus. This is achieved by selectively performing consensus a posteriori, i.e., only when absolutely necessary. Our on-demand consensus comes at the price of restricting the Byzantine participants to be less than a one-fifth minority in the system, which is the optimal threshold.

We formally prove the correctness of our system and present an open-source implementation, which inherits many features from the Ethereum ecosystem.

\end{abstract}

\section{Introduction}
Following the famed white paper of Satoshi Nakamoto \cite{bitcoin}, a plethora of digital payment systems (cryptocurrencies) emerged. The basic functionality of such payment systems are money transfer transactions.
These transactions are stored in a distributed ledger, a fault-tolerant and cryptographically secured append-only database. 
Most cryptocurrencies have a ledger where transactions are \textit{totally ordered}, effectively forcing all participants of the system to perform the state transitions sequentially.
This sequential verification of all transactions is considered the main bottleneck of distributed ledger solutions \cite{MerelyBroadcastig}. 

However, in reality, most transactions have no dependencies between each other. For example, a transaction from Alice to Bob and a transaction from Charlie to Dani can be performed in any order.
Verifying such independent transactions in parallel offers a vast efficiency improvement.
Indeed,  recent research proposes \textit{``no-consensus''} payment systems that do not order independent transactions \cite{mathys2021limitlessly,MerelyBroadcastig}. Such systems can achieve unbounded transaction throughput, as all transactions can be verified in any order, in parallel.



However, no-consensus payment systems suffer from fundamental limitations, as they lack the means to deal with conflicting inputs: If Charlie sets up two transactions, one for Alice, one for Bob, but Charlie does not have enough funds for both transactions, no-consensus payment system might end up in a deadlock, with Charlie ultimately losing access to her account, and neither Alice nor Bob getting paid. The same problem fundamentally prevents no-consensus systems from supporting general smart contracts, where many uncoordinated parties might issue conflicting inputs to the same smart contract at the same time.

\begin{wrapfigure}{r}{0.5\textwidth}
\centering
\includegraphics[width=0.40\textwidth]{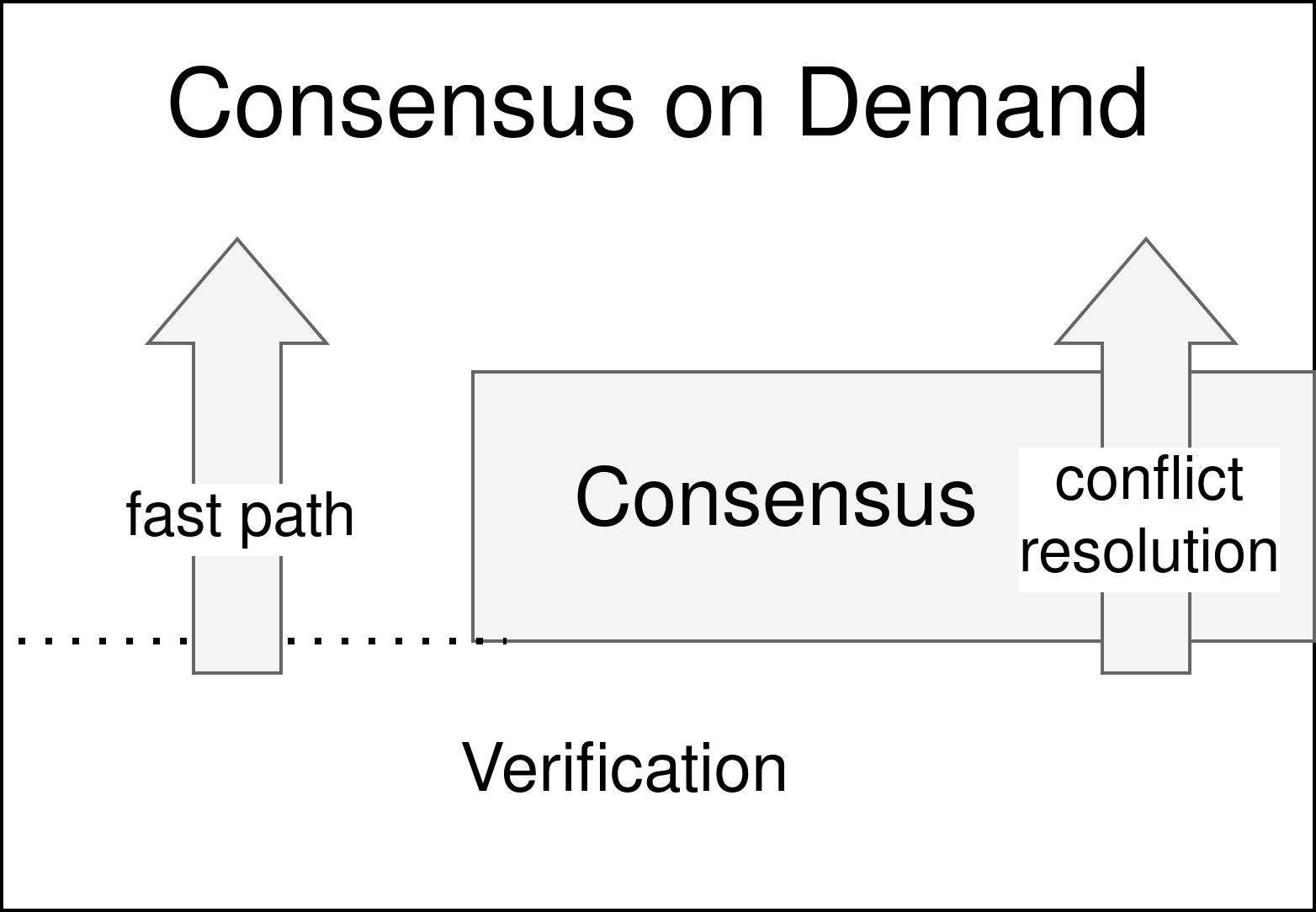}
\caption{A high level overview of our protocol.}
\label{fig:bigpicture}
\end{wrapfigure}

We are faced with a choice: either we use a total ordering currency which cannot scale to a high transaction throughput, or we use a parallel no-consensus verification system that is functionally restricted, and cannot resolve conflicting transactions.

In this work we propose a system which combines the advantages of both approaches. Our system offers all of the benefits of no-consensus systems, such as in principle unbounded throughput and powerful resiliency to network attacks. Our design first tries to verify every transaction without performing consensus. Only if a transaction cannot be verified on this ``fast path'', we invoke a consensus routine to resolve potential conflicts.

Our contributions are as follows:

\begin{enumerate}
    
    \item We present a protocol we call \apost. Assuming access to an existing consensus protocol, \apost is a wrapper algorithm where the first phase offers the benefits of no-consensus systems. In situations where conflicting inputs cannot be processed by pure no-consensus systems (and only those situations), \apost invokes the consensus instance to resolve the deadlock. The wrapper protocol is resilient to completely asynchronous network conditions as long as $n > 5f$, where $n$ is the total number of participants and $f$ is the number of Byzantine participants.
    The common case (no consensus) is optimal with regard to latency and does not rely on complex broadcast primitives. 
    
    Thus, we combine the power of processing unrelated inputs in parallel with the ability to resolve conflicting inputs when needed and pave the way for implementations of systems with unbounded throughput and full smart contract functionality.
    
    \item We exhibit our idea in the context of online payments. We describe our protocol, including the pseudocode, and prove the algorithm's correctness.

    \item We implement our design as a digital currency following a no-consensus approach enhanced with consensus on demand. A smart contract is used as the example consensus instance. Our implementation is built on top of the Ethereum client \textit{go-ethereum}, and thus features a network discovery protocol and advanced wallets, while being compatible with the Ethereum ecosystem.
    
    
    
    
\end{enumerate}

\section{Model}

We distinguish between clients and servers. Clients are free to enter and leave the system as they please. 
Servers are in charge of securing the system. We assume that the set of servers $\Pi$ is fixed and known to all servers.

Clients and servers that follow the protocol are said to be honest. Byzantine clients or servers are subject to arbitrary behavior and might collude when attempting to compromise the system's security.
We assume there are no more than $f$ Byzantine servers and that the set of Byzantine servers is static. Further, let $n = |\Pi|$. We assume that $n > 5f$, in other words, less than one-fifth of servers are Byzantine. 

Servers are connected all-to-all with authenticated links. Communication is {\bf asynchronous} i.e. messages are delivered with arbitrary delays. We assume standard cryptographic primitives to hold, more specifically, MACs and signatures cannot be forged.

Finally, the model might have to be restricted further in order to reflect the assumptions needed for the choice of the underlying consensus instance. For the consensus algorithm chosen for our implementation (see \cref{implementation}) we indeed assume synchronous communication.

\section{Problem Statement} \label{problemstatement}

We formulate the problem in the context of a cryptocurrency. Initially, the state of the system consists of a known assignment of currency amounts to clients.
The system's purpose is to accept transactions, where a transaction $t = (sender, sn, recipient, amount)$ moves an $amount$ of currency from a $sender$ to a $recipient$. Each client can issue transactions as the sender, where the sequence number $sn$ starts from 0 and increases by 1 for each transaction.

\begin{definition}
 Two transactions $t$ and $t'$ are said to \textbf{conflict}, if they have the same sender and sequence number but $t \neq t'$, i.e., the recipient or the amount differ.  
\end{definition}

Existing broadcast-based payment systems \cite{MerelyBroadcastig,auvolat2020money} provide the guarantees of a Byzantine reliable broadcast for every transaction:

\begin{definition}
\label{definition:rb}
Each honest server observes transactions from a set of conflicting transactions $\{t_0, t_1, . . . \}$. \textbf{Byzantine reliable broadcast} satisfies the following properties:

\begin{enumerate}
    \item \textbf{Totality}: If some honest server accepts a transaction, every honest server will eventually accept the same transaction.
    \item \textbf{Agreement}: No two honest servers accept conflicting transactions. 
    \item \textbf{Validity}: If every honest server observes the same transaction (there are no conflicting transactions), this transaction will be accepted by all honest servers.
\end{enumerate}
\end{definition}

The totality and agreement properties guarantee consistent state of the system and that at most one transaction per unique $(sender, sn)$ pair can be accepted, thus preventing double-spends. 

Validity ensures that if the client issued only one transaction for a given sequence number, the transaction will indeed be accepted. However, otherwise the definition does not guarantee termination. In other words, if the client issues conflicting transactions, the system might deadlock and never decide on accepting any of them.


Through the use of this weak abstraction, broadcast-based payment systems combine many benefits, such as resilience to complete asynchrony and fast acceptance. The standout advantage is perhaps the inherent ability to parallelize the processing of independent transactions, resulting in unbounded throughput through horizontal scaling \cite{baudet2020fastpay,mathys2021limitlessly}.

The crucial assumption that well-behaved clients will not issue conflicting transactions is warranted for a rudimentary payment system. However, it inherently precludes more advanced applications where conflicting inputs naturally occur, such as uncoordinated parties issuing conflicting inputs to a smart contract. To support the full range of blockchain applications, a stronger guarantee needs to hold:

\begin{definition}
\label{definition:consensus}
Each honest server observes transactions from a set of conflicting transactions $\{t_0, t_1, . . . \}$. \textbf{Consensus} satisfies the following properties:

\begin{enumerate}


    \item All properties of Byzantine reliable broadcast, and
    \item \textbf{Termination:} Every honest server eventually accepts one of the observed transactions.
    
\end{enumerate}
\end{definition}

The objective of this work is to combine the benefits of broadcast-based designs with the power of consensus:
\begin{inparaenum}[a)]
\item non-conflicting transactions are to be processed in a broadcast-based fashion: each honest server broadcasting one acknowledgement for a transaction is enough to accept it; and
\item consensus is supported to resolve conflicts.
\end{inparaenum}




\section{Related Work}
\label{relatedwork}

\renewcommand{\arraystretch}{1.5}
\begin{table*}[]\caption{A comparison of existing solutions and \apost (CoD). The CoD wrapping of a consensus is asynchronous and leaderless, and thus any potentially stronger assumptions are inherited from the consensus instance being used. }\label{tab:comparison}
\centering
\begin{tabularx}{\textwidth}{|c| >{\centering\arraybackslash}X| >{\centering\arraybackslash}X| >{\centering\arraybackslash}X| >{\centering\arraybackslash}X| >{\centering\arraybackslash}X| >{\centering\arraybackslash}X| >{\centering\arraybackslash}X| >{\centering\arraybackslash}X| >{\centering\arraybackslash}X| } \toprule
& \rot{80}{\parbox{2cm}{\centering Bitcoin and\\ Ethereum \cite{bitcoin}}} & \rot{80}{\parbox{2cm}{\centering Ouroboros \cite{kiayias2017ouroboros}}} & \rot{80}{\parbox{2cm}{\centering Algorand \cite{gilad2017algorand}}} & \rot{80}{\parbox{2cm}{\centering PBFT \cite{castro2002practical}}} & \rot{80}{\parbox{2cm}{\centering Red Belly \cite{redbelly}}} & \rot{80}{\parbox{2cm}{\centering BEAT \cite{beat}}} & \rot{80}{\parbox{2cm}{\centering Broadcast-based \cite{MerelyBroadcastig}}} & \rot{80}{\parbox{2cm}{\centering CoD + PBFT}} & \rot{80}{\parbox{2cm}{\centering CoD + BEAT}} \\ \midrule
Energy-efficient & & \checkmark & \checkmark &\checkmark &\checkmark & \checkmark & \checkmark & \checkmark & \checkmark \\ \hline
\makecell{Deterministic \\ finality} & & & \checkmark & \checkmark & \checkmark &\checkmark &\checkmark & \checkmark & \checkmark \\ \hline
Permissionless & \checkmark &\checkmark & \checkmark & & & & & & \\ \hline
Leaderless & & & & & \checkmark & \checkmark & \checkmark & & \checkmark \\ \hline
Asynchronous & & & & & & \checkmark & \checkmark & & \checkmark \\ \hline
Parallelizable & & & & & & & \checkmark & \checkmark & \checkmark \\ \hline
Consensus & \checkmark & \checkmark & \checkmark & \checkmark & \checkmark & \checkmark & & \checkmark & \checkmark \\ \bottomrule
\end{tabularx}
\end{table*}

\subsubsection{Broadcast-based Protocols} In 2016 Gupta \cite{gupta2016non} points out that a payment system does not require consensus. Later, Guerraoui et al. \cite{ConsensusNumber2019} prove that the consensus number of a cryptocurrency is indeed 1 in Herlihy’s hierarchy \cite{herlihy1991wait}.

Both Guerraoui et al. \cite{MerelyBroadcastig} and Baudet et al. \cite{baudet2020fastpay} propose a payment scheme where the ordering of transactions is purely determined by the transaction issuer. In their simplest form those currencies rely on Byzantine reliable broadcast,  as originally defined by Bracha and Toueg \cite{bracha}. Srikanth and Toueg \cite{Srikanth1987} as well as Bracha \cite{BRACHA1987130} propose well-known Byzantine reliable broadcast algorithms with $\mathcal{O}(n^2)$ message complexity per instance.
We use Bracha's Double-Echo algorithm \cite{BRACHA1987130} as a fundamental building block and comparison to our approach.

The Cascade protocol \cite{Sliwinski2021asynchronous} promises similar benefits, while being permissionless, i.e., participants are free to enter and leave the system as they please. 

Other approaches have proposed a probabilistic Byzantine reliable broadcast \cite{Extended}. By dropping determinism, efficiency is gained, more specifically $\mathcal{O}(n \thinspace \log(n))$ messages are shown to be sufficient for each transaction. Our implementation relies on a practical and widely adopted probabilistic broadcast protocol.

Instead, it is possible to drop the \textit{totality} property of Byzantine reliable broadcast and build a payment system where servers distribute themselves proof (a list of signatures) that they are indeed in the possession of the claimed funds. This was also proposed by Guerraoui et al. \cite{MerelyBroadcastig}, based on a digital signature approach inspired by Malkhi and Reiter \cite{malkhi1997high}. The message complexity is hereby improved to $\mathcal{O}(n)$.

\subsubsection{Remedying the Consensus Bottleneck} 
Early work by Pedone et al. \cite{pedone1999generic} and Lamport \cite{lamport2005generalized} recognizes that \textit{commuting} transactions do not need to be ordered in the traditional state machine replication (SMR) problem with crash failures. Follow-up protocols also deal with Byzantine faults and show fundamental lower-bounds \cite{lamport2006lower,pires2018generalized,raykov2011Byzantine}.

Removing global coordination in favor of weaker consistency properties also receives a lot of attention outside the area of state machine replication. Conflict-free Replicated Data Types (CRDT) \cite{rDTGotsman,CRDTOverview} provide a principled approach to performing concurrent operations optimistically, and have recently also been applied to permissioned blockchains \cite{FabricCRDT}. 



It is often tricky to compare protocols, as they can differentiate themselves in one of the many dimensions, such as synchrony, fault-tolerance and fast path latency \cite{bazzi2021clairvoyant}. A recent protocol called Byblos \cite{bazzi2021clairvoyant} achieves 5-step latency in a partially synchronous network when $n > 4f$. Suri-Payer et al. \cite{suri2021basil} improve the fast path latency to 2 communication steps, in the absence of Byzantine behavior.

Kursawe's optimistic Byzantine agreement protocol \cite{kursawe2002optimistic} features a fast path paired with a consensus protocol in the slow path, with each component being modular. While Kursawe's proposed fast path requires synchronous rounds and no Byzantine failures to happen, our protocol features the same optimal fast path of a single round-trip, while not relying on synchrony and tolerating $f$ Byzantine servers. This comes at the cost of requiring $n \geq 5f + 1$ servers. This bound has been shown to be optimal by Martin et al. \cite{martin2006fast}. Kuznetsov et al. \cite{kuznetsov2021revisiting} have recently shown the lower bound to be $n \geq 5f - 1$ in the special case where the set of \textit{proposers} (clients) is a subset of \textit{acceptors} (servers). Their insight is to disregard the acknowledgement of a provably misbehaving server. Although we do not assume the required special case, as in our model the set of clients is external to servers and changing freely, the assumption might well be warranted in other contexts, wherein their approach is applicable to our work.

Our protocol improves upon the solutions of Kursawe and Kuznetsov et al. by being leaderless and asynchronous even in the slow path. This is crucial as leader-based protocols have been shown to be susceptible to throughput degradation in the case of even one slow replica \cite{underAttack,leaderlessConsensus,alvisi,beat,yin2019hotstuff}. Song et al. \cite{song2008bosco} solution is probably most similar to ours, as their Bosco algorithm provides the same decision latency as ours. However, their solution does not focus on reducing the number of invocations of the underlying consensus, meaning that consensus is still performed for every decision.

Sharding is the process of splitting a blockchain architecture into multiple chains, allowing parallelization as each chain solves the state replication task separately. The improvement brought forward in this area \cite{avarikioti2019divide} is orthogonal to the one we address in this work. Indeed, while having multiple shards allows systems to parallelize operations overall, inside each shard transactions still need to be processed sequentially. 


\subsubsection{Implementations} Recent systems that remove or reduce the need for consensus have shown great promise in terms of practical scalability. More specifically, Astro \cite{MerelyBroadcastig} is able to perform 20,000 transactions per second, in a network of 200 nodes, with transactions having a latency of less than a second. A similar system by Spiegelman et al. \cite{tusk} that uses consensus without creating overhead achieves 160,000 tx/sec with about 3 seconds latency in a WAN. The Accept system \cite{mathys2021limitlessly} scales linearly, and has been shown to achieve 1.5 Million tx/sec.

\section{A Simple Payment System}
\label{paymentsystem}


We describe a digital currency called \bcoin that serves as a foundation. The protocol disseminates transactions through separate instances of Byzantine reliable broadcast. Crucially, the protocol does not rely on transactions being executed sequentially.

As explained in Section \ref{problemstatement}, clients start with a given account balance. Clients can access a server to submit transactions $t = (sender, sn, recipient, amount)$.
We assume that all transactions are signed using public-key cryptography and that servers only handle transactions with valid signatures. Clients can go offline whenever they please, but are required to keep track of the number of transactions they have performed so far, in order to choose correct, i.e. increasing, sequence numbers.

The \bcoin algorithm determines the agreed order of transactions of a given client to be executed. A transaction accepted by the underlying Byzantine reliable broadcast instance is executed (i.e., the funds are moved) as soon as all previous transactions belonging to the corresponding sender are executed, and enough funds are available in the sender's balance.

The \textbf{\bcoin} interface of a server (\textit{bc}) exports the following events:
\begin{itemize}
    \item \textbf{Request}: $\event{bc}{Transfer}[s, sn, r, a]$ : Allows a client $s$ to submit a transaction with sequence number $sn$ sending $a$ units of cryptocurrency to a recipient client $r$. 
    \item \textbf{Request}: $\event{bc}{RequestBalance}[c]$ : Retrieves the amount of cryptocurrency client $c$ currently owns.
    \item \textbf{Indication}: $\event{bc}{Balance}[c, a]$: Amount $a$ of cryptocurrency currently owned by client $c$.
\end{itemize}


\begin{algorithm}[hbt!]
\begin{algorithmic}[1]

\Uses
    \Instance{Authenticated Perfect Point-to-Point Links}{al}
    \Instance{Byzantine Reliable Broadcast}{rb}
\EndUses

\Upon{bc}{Init}{initialDistribution}
    \State $currentSN := [\,](-1)$;                 \algorithmiccomment{dictionary initialized with -1}
    \State $pending := \{\}$;                   \algorithmiccomment{empty set}
    \State $balance := initialDistribution$;       \algorithmiccomment{dictionary}

\EndUpon

\Upon{bc}{RequestBalance}{client}
    \Trigger{bc}{Balance}{client, \textit{balance}[client]};
\EndUpon

\Upon{bc}{Transfer}{[sender, sn, recipient, amount]}
    
    \State $t := [sender, sn, recipient, amount]$;
    \Trigger{rb}{Broadcast}{[sender, sn], t}; \algorithmiccomment{will be changed in \cref{aposteriori}}
\EndUpon

\Upon{rb}{Deliver}{[sender, sn], t}

    \State $pending[t.sender] = pending[t.sender] \cup t$;
\EndUpon

\UponCondition{ {\exists t \in pending} \textbf{ such that } isValidToExecute(t)}
    \State $balance[t.sender] = balance[t.sender] -t.amount$;
    \State $balance[t.recipient] = balance[t.recipient] + t.amount$;
    \State $currentSN[t.sender] = currentSN[t.sender] + 1$;
    \State $pending[t.sender] = pending[t.sender] \setminus t $; 
\EndUponCondition


\Procedure{isValidToExecute}{t} \label{line:isValid}
      \State \textbf{return} $currentSN[t.sender] = t.sn - 1 \land balance[t.sender] \geq t.amount$;
\EndProcedure

\end{algorithmic}
\caption{BroadcastCoin}
\label{algorithm:bcoin}
\end{algorithm}


In Byzantine reliable broadcast algorithms, a transaction $t$ typically undergoes the following steps before being accepted:
\begin{enumerate}
    \item \textit{Dissemination}: A server broadcasts $t$ received by a client by sending it to all servers.
    \item \textit{Verification}: Servers acknowledge $t$ if they have never acknowledged a conflicting transaction $t'$. 
    \item \textit{Approval}: Servers that receive more than $\frac{n+f}{2}$ acknowledgements for a transaction, broadcast an {\tt APPROVE} message. Servers also broadcast an {\tt APPROVE} message, if they see more than $f+1$ approvals. A server that receives more than $2f + 1$ approvals, accepts the transaction.
\end{enumerate}




\section{Consensus on Demand}
\label{aposteriori}

This section presents the core of our contribution that improves upon \bcoin by providing higher functionality as well as lower latency in the fast path. The Byzantine reliable broadcast instance $rb$ is substituted by two steps. 
A best-effort broadcast primitive is used to disseminate transactions efficiently. Then the first transaction $t$ for a given $(sender, sn)$ received by a server is the input value proposed in the corresponding \apost instance. \apost uses an underlying consensus instance to provide conflict resolution when necessary. We stress that the combination of the broadcast and consensus steps can be implemented in a variety of ways. The version we present in the following consists of best-effort broadcast paired with consensus as defined in \cref{definition:consensus}, while in \cref{discussion} we mention a different combination. As before, a transaction traverses three stages:

\begin{enumerate}
    \item \textit{Dissemination:} The transaction is broadcast to all servers. 
    \item \textit{Verification:} Servers issue an acknowledgement for the first valid transaction they observe for a given $(sender, sn)$ pair. If at any point, a server observes a quorum of more than $\frac{n+3f}{2}$ acknowledgements for a transaction $t$, the server accepts $t$.
    \item \textit{Consensus (opt.):} If after receiving $n-f$ acknowledgements servers observe conflicting acknowledgments, they propose the transaction for which they have observed the most acknowledgements up to this point to the consensus instance identified by the $(sender, sn)$ pair. The transaction decided by the consensus routine is then accepted immediately, if the transaction hasn't already been accepted by the fast path.
\end{enumerate}

Note that the first stage is identical to the first stage in the Byzantine reliable broadcast considered in \cref{paymentsystem}. Although the acceptance condition is also similar, it is performed without the additional broadcast round of {\tt APPROVE} messages. This means that in the common case, transactions are accepted with less delay. The final stage consists of performing consensus if necessary.

The crux of this construction is that a transaction accepted by the fast path should never conflict with a transaction accepted in the slow path. This holds true, since if a transaction $t$ can be accepted by an honest server in the fast path, even though conflicting transactions exists, then every other honest server is guaranteed to observe a majority of acknowledgements for $t$ in a quorum of size $n-f$. Thus, all honest servers will propose $t$ to the underlying consensus instance, and by its validity property, every server will eventually also accept $t$.

Fig. \ref{figure:intuition} illustrates this argument in the case where $n = 5f + 1$. There are $3f+1$ honest servers that acknowledge $t$ and $f$ honest servers that acknowledge $t'$. By issuing acknowledgements for $t$, the adversary could bring some servers to accept the transaction $t$ in the fast path. Hence, \apost should never accept $t'$. This can be guaranteed, as every quorum containing more than $n-f$ servers (such as $Q_1$) has a majority of servers acknowledging $t$. Thus, every server will propose $t$ to the consensus instance, which will accept $t$ due to its validity property. \cref{theorem:agreement+} proves this intuition.

\begin{figure}[h]
\includegraphics[width=.5\textwidth]{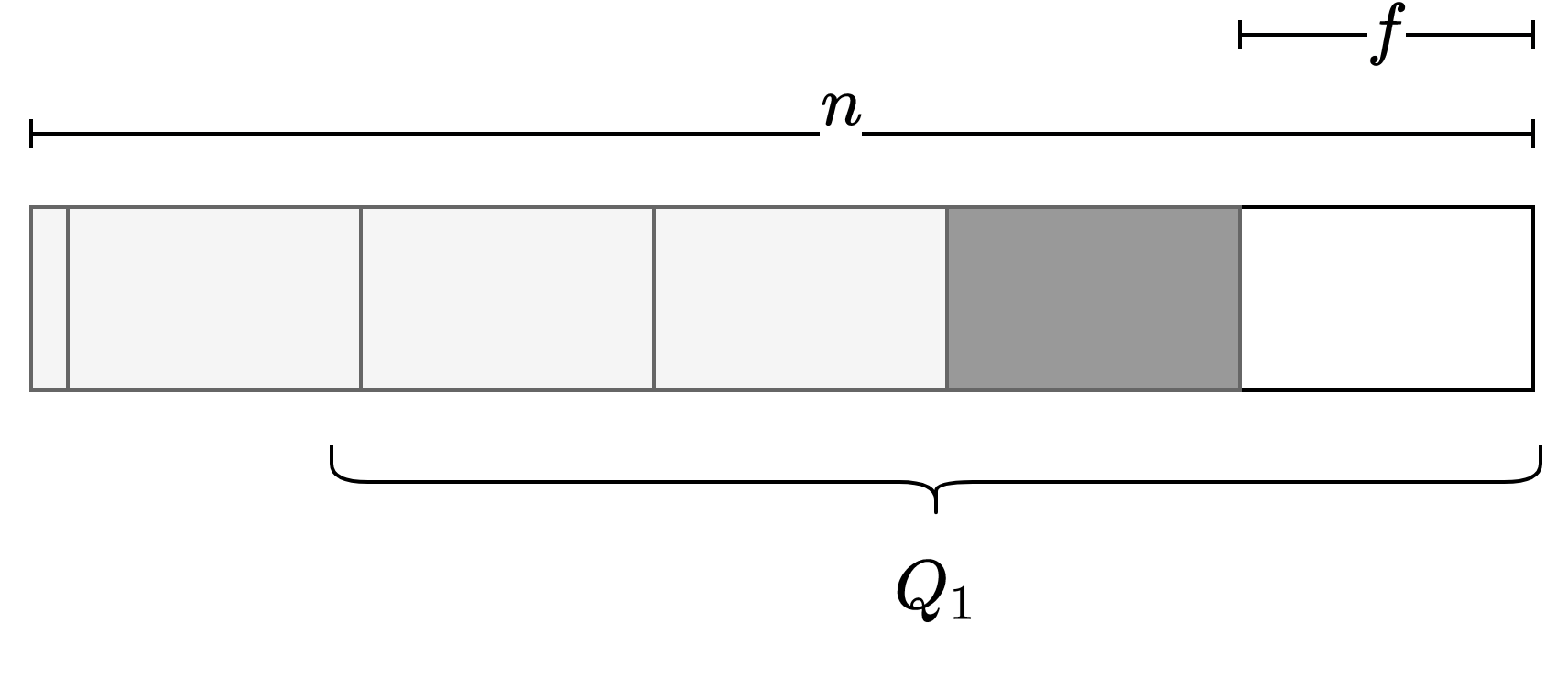}
\centering
\caption{The two shades of gray represent the share of honest servers acknowledging $t$ (light gray) and $t'$ (dark gray). The adversary is depicted in white, and can acknowledge either transaction. While a server might see more than $4f$ acknowledgments for $t$, no server sees a majority of acknowledgements for $t'$ in a quorum of $n-f$ servers. }

\label{figure:intuition}
\end{figure}





\begin{algorithm}[hbt!]
\begin{algorithmic}[1]
\Implements
    \Instance{Consensus}{fc}{\text{ for the (sender, sn) tuple }}
\EndImplements

\Uses
    \Instance{Consensus}{con}
    \Instance{Authenticated Perfect Point-to-Point Links}{al}
\EndUses

\UponPure{fc}{Init}
    \State $accepted, con\_proposed := \text{\tt False}$;
    \State $acks := [n](\bot)$;  \algorithmiccomment{array of size $n$ initialized with $\bot$}
\EndUponPure


\Upon{fc}{Propose}{t}
    \ForAll{q}{\Pi}
        \Trigger{al}{Send}{q, [\text{\tt ACK}, t]}
    \EndForAll
\EndUpon

\Upon{al}{Deliver}{p, [\text{\tt ACK}, t]}
    \If{acks[p] = \bot}
        \State $acks[\textit{p}] := t$;
    \EndIf
\EndUpon

\UponExists{t \neq \bot}{\text{\#}(\{p \in \Pi \mid acks[p] = t\}) \geq \frac{n+ 3f}{2} \text{ \textbf{and} } accepted = \text{\tt False}}
    \State $accepted := \text{\tt True}$;
    \Trigger{fc}{Accept}{t};
\EndUponExists

\UponExists{p, q \in \Pi}{acks[p] \neq acks[q] 
\text{ \textbf{and} } \text{\#}(\{p \in \Pi \mid acks[p] \neq \bot \}) \geq n-f 
\text{ \textbf{and} }  con\_proposed = \text{\tt False}}
    \State $majority := argmax_{t \in T}(\text{\#}(\{p \in \Pi \mid acks[p] = t\})$
    \State $con\_proposed := \text{\tt True}$;
    \Trigger{con}{Propose}{majority}
\EndUponExists

\UponSuchThat{con}{Accept}{t}{accepted = \text{\tt False}} 
    \State $accepted := \text{\tt True}$;
    \Trigger{fc}{Accept}{t};
\EndUponSuchThat


\end{algorithmic}
\caption{\apost}
\label{algorithm:apost}
\end{algorithm}



\begin{theorem}
\label{theorem:validity+}
\apost satisfies \textit{Validity}.
\end{theorem}
\begin{proof}
If every honest server observes the same transaction $t$, then every honest server broadcasts an acknowledgment for $t$. Thus every server is guaranteed to eventually observe at least $n-f$ acknowledgements for $t$. 
Since $f < \frac{n}{5}$, it follows that $\frac{n + 3f}{2} < n - f$, thus every server eventually accepts $t$.
\end{proof}

\begin{theorem}
\label{theorem:termination+}
\apost satisfies \textit{Termination}.
\end{theorem}
\begin{proof}
If every honest server observes the same transaction $t$, by the same argument as in \cref{theorem:validity+}, every server accepts $t$ in a single message round-trip. If an honest server observes too many conflicting acknowledgments to accept a transaction on the fast path, then at least two honest servers have issued conflicting transactions. Hence, eventually, every correct server will propose a transaction to the consensus instance $con$. By \textit{termination} of consensus, $con$ will eventually accept a transaction, and thus every honest server will eventually accept a transaction.
\end{proof}

\begin{theorem}
\apost satisfies \textit{Agreement}.
\label{theorem:agreement+}
\end{theorem}
\begin{proof}

First, let us assume that a server accepts a transaction $t$ without using the consensus instance. This means that the server has seen more than $\frac{n+3f}{2}$ acknowledgments for $t$. This implies that more than $\frac{n+3f}{2} - f = \frac{n+f}{2}$ honest servers have acknowledged $t$.

Before proposing, every server waits for the arrival of $n-f$ acknowledgements, out of which at least $n-2f$ come from honest servers. Together, both sets contain  more than $\frac{n+f}{2} + n - 2f = \frac{3(n - f)}{2}$ acknowledgements coming from honest servers. However, there are no more than $n-f$ honest servers, meaning that both sets have more than $\frac{3(n - f)}{2} - (n-f) = \frac{n-f}{2}$ acknowledgements in common. This implies that acknowledgements for $t$ will be the most received acknowledgement at every honest server. 

Therefore, every honest server will either accept $t$ though its fast path or, if there is a conflicting transaction, propose $t$ to the consensus instance. Due to its validity property, no honest server will accept a value different from $t$, thus satisfying agreement.

If no server observes more than $\frac{n+3f}{2}$ acknowledgments for a single transaction, then all honest servers will fall back to the consensus instance, and due to its agreement property, the agreement of \apost is also satisfied. 
\end{proof}

\section{Discussion}
\label{discussion}

\paragraph{\bf Throughput.}
No-consensus payment systems have been shown to scale linearly with more computing resources \cite{MerelyBroadcastig,baudet2020fastpay}. In particular the simple design of Mathys et al. \cite{mathys2021limitlessly} can be directly applied as the implementation of the fast path of our design, and their result supports our claim that the fast path of our protocol has in principle unbounded throughput.

\paragraph{\bf Slow path abuse.}
In \apost, malicious clients can increase the likelihood that consensus needs to be performed by submitting conflicting transactions intentionally.

Due to the completely asynchronous communication model, in our protocol servers keep listening for potentially conflicting acknowledgements of past transactions that might trigger a consensus invocation. This requirement can be avoided by replacing best-effort broadcast in the fast path with (probabilistic) reliable broadcast. In this configuration, servers for which the fast path succeeds do not have to participate in the slow path at all, as thanks to reliable broadcast's totality property, every honest server is guaranteed to eventually be able to complete the fast path. This modification would make it harder for malicious clients to intentionally invoke consensus, but on the other hand a more complicated broadcast primitive would be used (two echo rounds instead of one).

Intentional abuse of the slow path can also be addressed through game-theoretic means. Economic incentives, such as fees, can be set up so that intentional consensus invocation is adequately costly for a malicious client. The subject of incentive schemes in blockchain systems is broad, as different aspects of the system's functionality need to be considered depending on the application. It is thus left outside the scope of this paper.



\paragraph{\bf Fast path-only synchronization.}

We presented \apost in the form where the consensus outcome is accepted by the servers without further steps. Consider the following addition to our protocol. Suppose a server $s$ has not acknowledged a transaction $t$ in the fast path, and later $t$ is the result of consensus. Even though $s$ might have acknowledged a conflicting transaction $t'$ in the fast path, let now $s$ broadcast a fast path acknowledgement for $t$. By introducing this rule, we ensure that all honest servers that observe the consensus outcome additionally issue a fast path acknowledgement. Afterwards, all accepted transactions can be determined only following the fast path condition.

In this setup, any records of consensus performed by the system can be forgotten, as any agent synchronizing with the state of the system conveniently only needs to be supplied with the fast path acknowledgements.

\section{Implementation}
\label{implementation}

We implement the \bcoin protocol described in \cref{algorithm:bcoin} by utilizing the core of the \textit{go-ethereum} client for Ethereum. 
The main modules that are of relevance are briefly described below.


\textbf{Transactions:} There are two types of transactions in Ethereum. We only support transactions that lead to message calls, and do not support transactions that lead to contract creation. Transactions are broadcast using the Ethereum Wire Protocol \cite{ethprotocol} that probabilistically disseminates blocks through gossip with a sample size of $\sqrt{n}$.

\textbf{Blocks:} The fundamental building blocks of Ethereum also lay at the core of our protocol. However, instead of using a single chain of blocks to totally order transactions, blocks are used to broadcast batches of acknowledgments. This is done by including all transactions that should be signed in a block created by the server. The block signature proves the authenticity of all acknowledgments. The \textit{parentHash} field of a block is also kept, in order to refer to the previous block, which allows for easier \textit{synchronization} between servers.

\textbf{Blockchain:} As every server issues its own chain of blocks, we re-purpose the blockchain abstraction to keep track of all chains in a DAG and allow for synchronization with new clients in future extensions.

\textbf{Mining:} We replace the proof-of-work engine with our own logic that determines which transactions from the transaction pool are safe to be acknowledged. Acknowledgements are batched in blocks, signed and broadcast every 5 seconds. 

\textbf{Transaction pool:} The transaction pool module is managing new transactions in Ethereum. Most functions and data structures shown in the pseudocode of \cref{algorithm:bcoin} are closely matching the implementation of this module.


We complete the implementation of our protocol by enhancing the no-consensus payment system with the \apost algorithm. We do so by plugging in a simple consensus algorithm built on the Ethereum Rinkeby testnet. More specifically, we provide a smart contract that is able to perform consensus for any $(sender, sn)$ instance. The contract terminates either when $f + 1$ equal proposals for $t$ are collected, in which case it immediately accepts $t$. Alternatively, once $2f+ 1$ proposals are collected, the contract accepts the most frequent input. The smart contract is called \textit{Multishot} and its implementation can be found in \cite{githubMultishot2021}, while Appendix \ref{appendix:multi} shows the pseudocode and the correctness proof of this algorithm.

While our algorithm is agnostic to the underlying consensus algorithm used, this simple smart contract allows us to demonstrate the effectiveness of \apost, while keeping our implementation in the Ethereum ecosystem. 

These few modules make up most of the changes that were required for us to leverage a large part of the existing \textit{go-ethereum} infrastructure. This allows us to take advantage of the network discovery protocol \cite{ethprotocol} and the support for hardware wallets. Moreover, our server can simultaneously function as a client, which can be controlled through a variety of interfaces. While the regular JavaScript console can be used, the client can also be addressed via a standard web3 JSON-RPC API accessible through HTTP, WebSockets and Unix Domain Sockets. The complete infrastructure is open source \cite{github2021}.

\bibliographystyle{splncs04}
\bibliography{bibliography}

\newpage

\begin{subappendices}
\renewcommand{\thesection}{\Alph{section}}%

\section{Implementation Details}

We visualize the core pipeline of our implementation in \cref{fig:implementation}. The functions and data structures shown closely match our changes to the \textit{go-ethereum} transaction pool module, where most of the logic regarding the acceptance of a transaction is situated. 

\begin{figure*}[ht]
\centering
\includegraphics[width=0.9\textwidth]{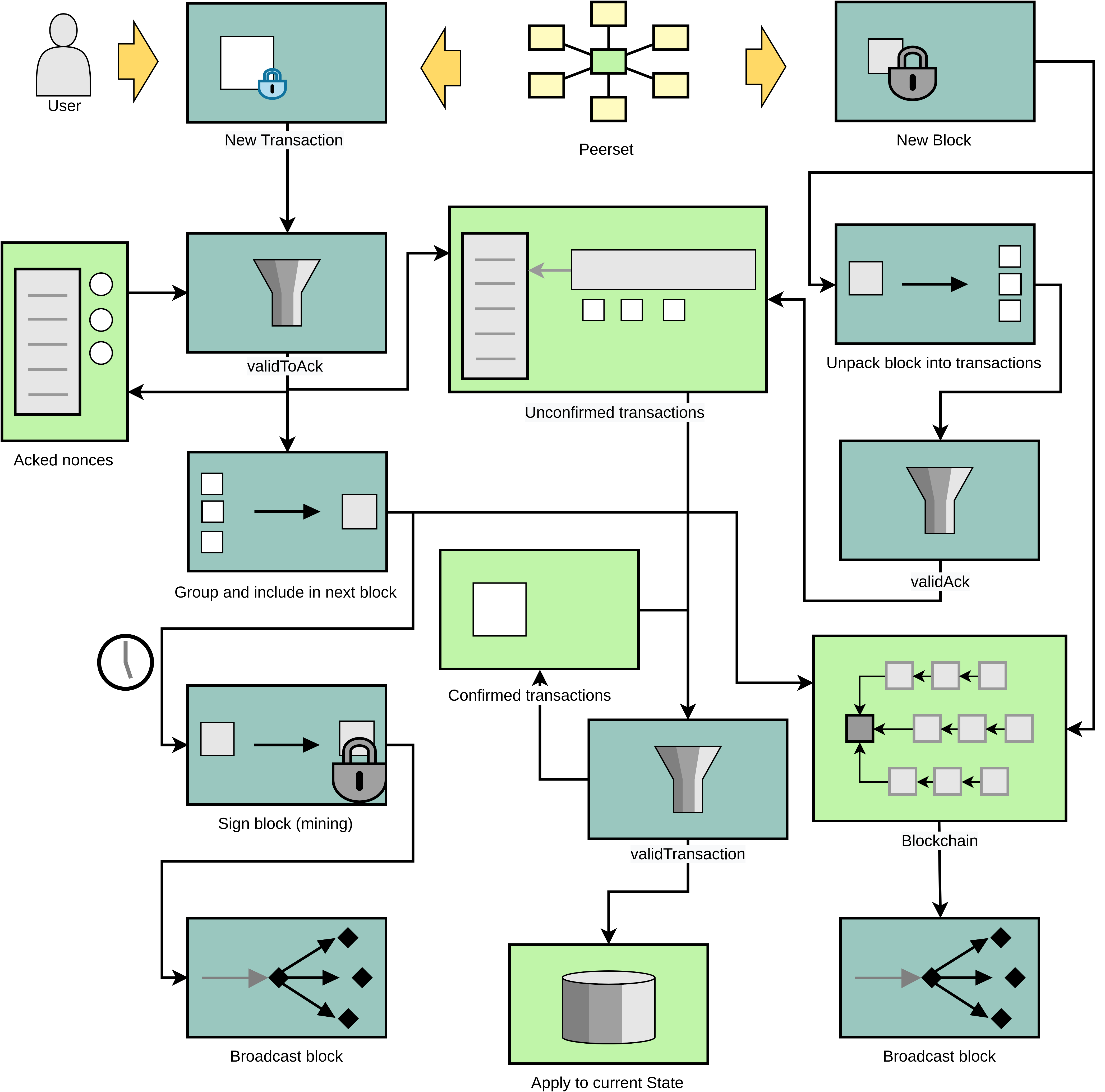}
\caption{The core pipeline of our server implementation. The lighter rectangles represent data structures, while the dark rectangles are network or computation operations.}
\label{fig:implementation}
\end{figure*}

As an underlying source of consensus an Ethereum smart contract is used. Written in Solidity, its pseudocode for a single instance is described by \cref{algorithm:multishot}.

We show that the contract satisfies consensus under the assumption that the Ethereum blockchain itself does not revert and indeed itself provides consensus.

\textbf{Agreement} is trivially satisfied through the use of a designated leader (smart contract).

\textbf{Termination} is satisfied since if all $n \geq 4f + 1$ honest servers propose, there are guaranteed to be $2f + 1$ proposals, upon which the algorithm terminates.

\textbf{Validity} is satisfied, since if all honest servers propose the same transaction $t$, then every sample of $f + 1$ proposals contains at least one proposal for $t$. Further, no matter what the majority proposal is, either it was proposed more than $f$ times, in which case validity is satisfied through the previous argument, or, more than $f$ proposals were not $t$, in which case two honest servers have proposed different values. Hence, we are guaranteed to satisfy validity by accepting the majority value.

\label{appendix:multi}

\begin{algorithm}[ht]
\begin{algorithmic}[1]
\Implements
    \Instance{Consensus}{con}{\text{ for the (sender, sn) tuple }} 
\EndImplements

\UponPure{con}{Init}
    \State \textit{proposals} := $[n](\bot)$; \algorithmiccomment {Array of size n.}
    \State \textit{accepted} := \texttt{False};
    
\EndUponPure

\Upon{con}{Propose}{u, [sender, sn], t}
    \State \textit{proposals}[u] := t;
\EndUpon

\UponExists{t, sender, sn \neq \bot}{\text{\#}(\{u \in \Pi \mid proposals[u] = t\}) \geq f + 1 \text{ \textbf{and} } accepted = \text{\tt False}}
    \State \textit{accepted} := \text{\tt True};
    \Trigger{con}{Accept}{[sender, sn],t}
\EndUponExists

\UponExists{sender, sn \neq \bot}{\text{\#}(\{u \in \Pi \mid proposals[u] \neq \bot) \geq 2f + 1 \text{ \textbf{and} } accepted = \text{\tt False}}
    \State \textit{accepted} := \text{\tt True};
    \State t := $argmax_{t \in T}(\text{\#}(\{u \in \Pi \mid proposals[u] = t\})$
    \Trigger{con}{Accept}{[sender, sn],t}
\EndUponExists

\end{algorithmic}
\caption{MultishotConsensus}
\label{algorithm:multishot}
\end{algorithm}

\end{subappendices}

\end{document}